\documentclass[11pt]{article}%
\usepackage{amsmath}
\usepackage{amsfonts}
\usepackage{amssymb}
\usepackage{graphicx}
\usepackage{color}%
\providecommand{\U}[1]{\protect\rule{.1in}{.1in}}
\newtheorem{theorem}{Theorem}[section]

\newtheorem{corollary}{Corollary}[section]

\newtheorem{definition}{Definition}[section]
\newtheorem{example}{Example}[section]

\newtheorem{lemma}{Lemma}[section]

\newtheorem{remark}{Remark}[section]

\newenvironment{proof}[1][Proof]{\noindent\textbf{#1.} }{\hfill \rule{0.7em}{0.7em} \medskip}
\numberwithin{equation}{section}

\newcommand{\drm}{\mathrm{d}}
\newcommand{\sig}{\sigma}

\begin{document}

\title{Optimality of Excess-Loss Reinsurance \break under a Mean-Variance Criterion \break REPAIR OF SECTION 3.3}

\author{Danping Li\thanks{Department of Statistics and Actuarial Science, University
of Waterloo, Waterloo, ON, N2L 3G1, Canada (d268li@uwaterloo.ca)}
\and Dongchen Li\thanks{Department of Statistics and Actuarial Science, University
of Waterloo, Waterloo, ON, N2L 3G1, Canada (d65li@uwaterloo.ca)}
\and Virginia R. Young \thanks{Corresponding author. Department of Mathematics, University of Michigan,
Ann Arbor, MI, 48109, USA (vryoung@umich.edu)}}

\date{\small \today}

\maketitle

\begin{abstract}
In this paper, we study an insurer's reinsurance-investment problem under a mean-variance criterion. We show that excess-loss is the unique equilibrium reinsurance strategy under a spectrally negative L\'{e}vy insurance model when the reinsurance premium is computed according to the expected value premium principle. Furthermore, we obtain the explicit equilibrium reinsurance-investment strategy by solving the extended Hamilton-Jacobi-Bellman equation.

\smallskip


\smallskip


\end{abstract}


\baselineskip15pt

\section{Introduction}

An integrated reinsurance and investment strategy is commonly employed by an insurer (cedent) to increase its underwriting capacity, stabilize underwriting results, protect itself against catastrophic losses, and achieve financial growth. The study of an insurer's optimal reinsurance-investment strategy has received considerable attention in the literature of actuarial science under a
variety of criteria, including minimizing the probability of ruin (see, for example, Promislow and Young \cite{PY05}, Zhang et al.\ \cite{ZZG07}, and Chen et al.\ \cite{CLL10}), maximizing the expected utility of terminal wealth (see, for example, Liu and Ma \cite{LM09}, Bai and Guo \cite{BG10}, Gu et al.\ \cite{GGLZ12}, and Liang and Bayraktar \cite{LB14}), and maximizing expected terminal wealth subject to a constraint on the variance, the so-called \textit{mean-variance criterion} (see, for example, B\"{a}uerle \cite{B05} and Zeng and Li \cite{ZL11}).

The mean-variance criterion is closely related to maximizing expected utility of terminal wealth. Indeed, Pratt \cite{P64} observes that the certainty equivalence for a ``small'' random gain $Y$ under expected utility theory approximately equals
\begin{equation}
\label{eq:mv}\mathbb{E}(Y) - \frac{\gamma}{2} \, \mathrm{Var}(Y),
\end{equation}
in which $\gamma$ is the absolute risk aversion of the utility maximizer. Note that maximizing \eqref{eq:mv} is precisely the mean-variance criterion. Also, under fairly general conditions, optimal insurance is deductible insurance for a risk-averse utility maximizer (see, for example, Arrow \cite{A63}, van Heervarden \cite{VH91}, and Moore and Young \cite{MY06}). Thus, when maximizing \eqref{eq:mv} (or solving a related game) with $Y$ equal to terminal wealth of an insurance company, we expect that optimal (or equilibrium) reinsurance will be deductible, or excess-loss, reinsurance, which we prove below in Theorem \ref{thm:2}.  Furthermore, because the risk aversion $\gamma$ is constant, the deductible is independent of the surplus of the insurer.

Under the mean-variance criterion, the reinsurance-investment problem is time-inconsistent in the sense that Bellman's optimality principle fails. To tackle the time inconsistency, we formulate the problem as a non-cooperative game and solve for a subgame perfect Nash equilibrium. Specifically, at every time point, the player solves for an \textit{equilibrium strategy} by treating the problem as a game against all future versions of himself. An equilibrium strategy is, thus, time-consistent. One can trace this approach to Strotz \cite{S55}, and it has recently been further developed by Bj\"{o}rk and Murgoci \cite{BM10} for a general class of objective functions in a Markovian
framework. Due to the importance of time consistency for a rational insurer, the approach has already been applied by many authors to solve for equilibrium strategies in the literature of reinsurance-investment problems (see, for example, Zeng et al.\ \cite{ZLL13} and Lin and Qian \cite{LQ15}).

Two types of reinsurance policies are most commonly studied in the literature on equilibrium reinsurance and investment under a mean-variance criterion: (1) proportional (quota-share) reinsurance (see, for example, Zeng and Li \cite{ZL11}, Shen and Zeng \cite{SZ15}, and the two references given at the end of the previous paragraph) and (2) excess-loss reinsurance (see, for example, Li et al.\ \cite{LRZ16}).  Given the rich literature, one question naturally arises which has not received much attention: \textit{Which reinsurance policy yields an equilibrium for an insurer under a mean-variance criterion among \textrm{all} reasonable reinsurance policies}? We show that buying excess-loss reinsurance is the unique equilibrium strategy under this criterion.


We model the insurer's basic surplus process, that is, the surplus process without any reinsurance-investment strategy, by a \textit{spectrally negative L\'{e}vy process}. The model is widely employed in the context of risk theory in the actuarial literature (see, for example, Yang and Zhang \cite{YZ01}, Chiu and Yin \cite{CY05}, Avram et al.\ \cite{APP07}, and Landriault et al.\ \cite{LRZ11}). It is a generalization of many insurance models studied in the context of reinsurance-investment problems, including the Brownian motion model (see, for example, Promislow and Young \cite{PY05}), the classical Cram\'{e}r-Lundberg model (see, for example, Zeng et al.\ \cite{ZLG16}), and the jump-diffusion model (see, for example, Zeng et al.\ \cite{ZLL13}).

We prove that, when the reinsurance premium is computed according to the expected value premium principle, excess-loss reinsurance is the unique equilibrium strategy for a time-consistent insurer under a mean-variance criterion. As mentioned above, this result is consistent with several in the literature; specifically, under the expected value premium principle and various objective
functions, excess-loss (re)insurance is optimal, including when maximizing the expected utility of terminal wealth (see, for example, Liang and Guo \cite{LG11}, Zeng and Luo \cite{ZL13}) and when minimizing the probability of ruin (see, for example, Zhang et al.\ \cite{ZZG07}, Meng and Zhang \cite{MZ10}, Bai et al.\ \cite{BCZ13}, and Zhou and Cai \cite{ZC14}).

The remainder of this paper is organized as follows. In Section 2, we formulate our model and define the equilibrium problem faced by the insurer.  In Section 3, we prove that excess-loss reinsurance is the unique equilibrium strategy, and we obtain explicit expressions for the equilibrium reinsurance-investment strategy and the corresponding equilibrium value function. We also discuss two problems closely related to the mean-variance criterion: (1) maximizing expected exponential utility of terminal wealth, and (2) maximizing the time-0 mean-variance criterion with commitment. In Section 4, we present some numerical examples to illustrate our findings, and Section 5 concludes the paper.

\section{Model formulation}

Let $\left(  \Omega, \mathcal{F}, \boldsymbol{F=} \left\{  \mathcal{F}_{t}\right\}_{t\geq0},\mathbb{P} \right)  $ be a filtered, complete probability space satisfying the usual conditions, and let $T > 0$ be a finite time horizon. Consider an insurer's basic surplus process modeled by a spectrally negative L\'{e}vy process defined on this probability space with dynamics
\[
\mathrm{d} U_{t} = c \, \mathrm{d} t + \sigma_{1} \, \mathrm{d} B_{t}^{(1)} - \int_{0}^{\infty}z \, N(\mathrm{d} z, \mathrm{d} t), \quad U_{0} > 0,
\]
in which $c > 0$ is the premium rate, $\sigma_{1} > 0$ is the volatility rate, $\big\{  B_{t}^{(1)} \big\}_{t \ge0}$ is an $\boldsymbol{F}$-adapted standard Brownian motion, and $N(\mathrm{d} z,\mathrm{d} t)$ is a Poisson random measure representing the number of insurance claims of size $(z, z + \mathrm{d} z)$ within the time period $(t, t + \mathrm{d} t)$. $B^{(1)}$ and $N$ are independent. For more information about L\'evy processes, please see Kyprianou \cite{K06}.

Denote the compensated measure of $N(\mathrm{d} z, \mathrm{d} t)$ by $\tilde{N}(\mathrm{d} z,\mathrm{d} t) = N(\mathrm{d} z,\mathrm{d} t) - \nu(\mathrm{d} z) \mathrm{d} t,$ in which $\nu$ is a L\'{e}vy measure such that $\int_{0}^{\infty}z \, \nu(\mathrm{d} z) < \infty$; $\nu(\mathrm{d} z)$ represents the expected number of insurance claims of size $(z, z + \mathrm{d} z)$ within a unit time interval. The insurer's premium $c$ is determined under the expected value principle, that is, $c = (1 + \theta) \int_{0}^{\infty}z \, \nu(\mathrm{d} z)$, in which $\theta> 0$ is the proportional safety loading of the insurer.

The insurer manages its insurance liabilities by purchasing a reinsurance policy (strategy) with retained claim $\left\{  \ell_{t} \right\}_{t \in [0, T]}$, with the only restriction $0 \le\ell_{t} \le Z_{t}$ when the claim equals $Z_{t}$ at time $t \in[0, T]$. Note that the reinsurer covers the excess loss $Z_{t} - \ell_{t}$. We will look for a reinsurance strategy given in feedback form by $\ell_{t} = \ell(Z_{t}, t)$, in which we slightly abuse notation by using $\ell$ on both sides of this equation. Technically, we should assume \textit{a priori} that the retention strategy depends on surplus, but in Theorem \ref{thm:2} below, we will find the equilibrium retention that is independent of the surplus. Thus, for simplicity, we omit $\ell$'s possible dependency on the surplus.

The time-$t$ premium rate of the reinsurance policy is given by
\[
(1 + \eta) \int_{0}^{\infty}(z - \ell(z, t)) \, \nu(\mathrm{d} z),
\]
determined again under the expected value principle, in which $\eta$ is the reinsurer's proportional safety loading. It is commonly assumed in the literature that $\eta> \theta$, indicating that a reinsurance policy is more expensive than the primary insurance, and by using this assumption, one generally avoids trivial results. Under the retention $\ell$, the dynamics of the surplus process is governed by
\begin{align*}
\mathrm{d} R_{t}  &  = \mathrm{d} U_{t} - (1 + \eta) \int_{0}^{\infty}[z - \ell(z, t)] \, \nu(\mathrm{d} z) \, \mathrm{d} t + \int_{0}^{\infty}[z -
\ell(z, t)] \, N(\mathrm{d} z, \mathrm{d} t)\\
&  = (1 + \theta) \int_{0}^{\infty}z \, \nu(\mathrm{d} z) \, \mathrm{d} t + \sigma_{1} \, \mathrm{d} B_{t}^{(1)} - (1 + \eta) \int_{0}^{\infty}[z -
\ell(z, t)] \, \nu(\mathrm{d} z) \, \mathrm{d} t\\
&  \quad- \int_{0}^{\infty}z \, N(\mathrm{d} z, \mathrm{d} t) + \int_{0}^{\infty}[z - \ell(z, t)] \, N(\mathrm{d} z, \mathrm{d} t)\\
&  = \int_{0}^{\infty}\left(  (\theta- \eta) z + \eta\ell(z, t) \right) \nu(\mathrm{d} z) \, \mathrm{d} t + \sigma_{1} \, \mathrm{d} B_{t}^{(1)} -
\int_{0}^{\infty}\ell(z, t) \, \tilde{N}(\mathrm{d} z, \mathrm{d} t).
\end{align*}

Furthermore, suppose the insurer invests in a financial market consisting of a risk-free asset with a constant interest rate $r > 0$ and a risky asset governed by a geometric Brownian motion with dynamics
\[
\mathrm{d}S_{t} = \mu\,S_{t}\,\mathrm{d}t + \sigma_{2}\,S_{t}\left( \rho\,\mathrm{d}B_{t}^{(1)} + \sqrt{1-\rho^{2}}\,\mathrm{d}B_{t}^{(2)}\right),\quad S_{0} > 0,
\]
in which $\mu>r$, $\sigma_{2}>0$, $\rho\in(-1,1)$, and $\big\{B_{t}^{(2)}\big\}_{t\geq0}$ is an $\boldsymbol{F}$-adapted standard Brownian motion, independent of both $B^{(1)}$ and $N$. Let $\pi_{t}$ denote the dollar amount of surplus invested in the risky asset at time $t$, and let $\left\{ X_{t}^{u}\right\}_{t\in\lbrack0,T]}$ denote the corresponding insurance surplus process under a \textit{reinsurance-investment strategy} $u:=\left( \ell(Z_{t},t),\pi_{t}\right)_{t\in\lbrack0,T]}$. The dynamics of the surplus process $\left\{  X_{t}^{u}\right\}_{t\in\lbrack0,T]}$ is, then, given by
\begin{align}\label{SDE3}
\mathrm{d}X_{t}^{u} &  =\pi_{t}\,\frac{\mathrm{d}S_{t}}{S_{t}}+\left( X_{t}^{u}-\pi_{t}\right)  r\,\mathrm{d}t+\mathrm{d}R_{t} \nonumber \\
&  =\left[  rX_{t}^{u}+(\mu-r)\pi_{t}+\int_{0}^{\infty}\left(  (\theta - \eta)z + \eta\ell(z,t)\right)  \nu(\mathrm{d}z)\right]  \mathrm{d}t\nonumber\\
&  \quad + \sqrt{\sigma_{1}^{2}+2\rho\sigma_{1}\sigma_{2}\pi_{t}+\sigma_{2}^{2} \pi_{t}^{2}} \,\mathrm{d}B_{t} - \int_{0}^{\infty} \ell(z,t)\,\tilde{N} (\mathrm{d}z,\mathrm{d}t),
\end{align}
in which $\{B_{t}\}_{t\geq0}$ is an $\boldsymbol{F}$-adapted standard Brownian motion, independent of $N$, defined by
\[
B_{t} = \frac{\sigma_{1} + \rho\sigma_{2}\pi_{t}}{\sqrt{\sigma_{1}^{2} + 2\rho \sigma_{1}\sigma_{2}\pi_{t} + \sigma_{2}^{2}\pi_{t}^{2}}}\,B_{t}^{(1)} + \frac{\sqrt{1-\rho^{2}}\,\sigma_{2}\pi_{t}}{\sqrt{\sigma_{1}^{2}+2\rho
\sigma_{1}\sigma_{2}\pi_{t}+\sigma_{2}^{2}\pi_{t}^{2}}}\,B_{t}^{(2)}.
\]

For ease of notation, let $\mathbb{E}_{x, t}\left[  \cdot \right] = \mathbb{E} \left[  \cdot \big| \, X_{t}^{u}=x \right]$ and $\mathrm{Var}_{x,t}\left[  \cdot \right] = \mathrm{Var} \left[  \cdot \big| \,X_{t}^{u}=x \right]$. 

\smallskip

\begin{definition}
\label{Def:1}
$($Admissible strategy$)$. A strategy $u = \left( \ell(Z_{t}, t), \pi_{t} \right)_{t \in[0,T]}$ is called \textrm{admissible} if it satisfies the following conditions:

\item {$($1$)$} $u$ is $\boldsymbol{F}$-progressively measurable;

\item {$($2$)$} For all $t \in[0, T]$ and $Z_{t} \ge0$, $0 \le \ell(Z_{t}, t) \le Z_{t}$;

\item {$($3$)$} For all $(x, t) \in\mathbb{R} \times[0, T]$, $\mathbb{E}_{x, t} \left[  \int_{t}^{T} (\ell^{2}(Z_{s}, s) +\pi_{s}^{2}) \, \mathrm{d} s \right]  < \infty$ with probability one;

\item {$($4$)$} For all $(x, t) \in\mathbb{R} \times[0, T]$, the stochastic
differential equation \eqref{SDE3} has a unique strong solution.


\end{definition}

The \textbf{main objective} of this paper is to study the reinsurance-investment problem for an insurer under a mean-variance criterion, that is, one who wishes to maximize $J^{u}(x,t)$, in which $J^{u}$ is given by
\begin{equation}
J^{u}(x,t) = \mathbb{E}_{x,t} \left[  X_{T}^{u}\right]  -\frac{\gamma}{2}\,\mathrm{Var}_{x,t}\left[  X_{T}^{u}\right]  ,\quad(x,t)\in\mathbb{R} \times [0, T],\label{Jdef}
\end{equation}
in which $\gamma > 0$ measures the insurer's degree of (absolute) risk aversion.

Maximizing $J^{u}(x, t)$ is a time-inconsistent problem in the sense that Bellman's optimality principle fails. We tackle the problem from a non-cooperative game point of view by defining an equilibrium strategy and its corresponding equilibrium value function; see, for example, Basak and Chabakauri \cite{BC10}, Bj\"ork and Murgoci \cite{BM10}, and Bj\"ork et al.\ \cite{BMZ14}.

\begin{definition}
\label{Def:2}
For an admissible strategy $u^{*} = \left(  \ell^{*}(Z_{t}, t), \pi^{*}_{t} \right)_{t\in[0, T]}$, for $\varepsilon> 0$, and for $t \in [0, T]$, define the strategy $u^{\varepsilon, t}$ by
\begin{equation}
\label{eq:uvareps}
u_{s}^{\varepsilon, t} = \left\{
\begin{array}
[c]{ll}
(\bar\ell(z, s), \bar{\pi}), & t \le s < t + \varepsilon, \vspace{0.1cm}\\
u^{*}_{s}, & 0 \le s < t \text{ or } t + \varepsilon\le s \le T,
\end{array}
\right.
\end{equation}
in which $\bar\ell(z, s)$ is an admissible retention strategy and $\bar{\pi}$ is a real constant. If, for all $(x, t) \in\mathbb{R} \times[0, T]$,
\[
\liminf_{\varepsilon\downarrow0} {\frac{J^{u^{*}}(x, t) - J^{u^{\varepsilon,
t}}(x, t)}{\varepsilon}} \ge0,
\]
then $u^{*}$ is an equilibrium strategy and $J^{u^{*}}(x, t)$ is the corresponding equilibrium value function.
\end{definition}

\section{Equilibrium reinsurance-investment strategy}

\subsection{Verification theorem}


We first provide a verification theorem whose proof we omit because it is similar to the proof of the verification theorem, Theorem 4.1, in Bj\"ork and Murgoci \cite{BM10}. Also, see the discussion in Bj\"ork et al.\ \cite{BMZ14} about applying the verification theorem under the mean-variance criterion.

For any admissible retention $\ell$ and for any constant $\pi\in\mathbb{R}$,
we define an integro-differential operator $\mathcal{A}^{\ell, \pi}$ as
follows:
\begin{align}
\label{IG}\mathcal{A}^{\ell, \pi} \, \phi(x, t)  &  := \lim_{\varepsilon
\downarrow0} \frac{\mathbb{E}_{x, t} \left[  \phi(X_{t+\varepsilon}^{u}, t +
\varepsilon) \right]  - \phi(x, t)}{\varepsilon}\nonumber\\
&  = \phi_{t}(x, t) + \left[  r x + (\mu- r)\pi+ \int_{0}^{\infty}\left(
(\theta- \eta)z + (1 + \eta) \ell(z, t) \right)  \nu(\mathrm{d} z) \right]
\phi_{x}(x, t)\nonumber\\
&  \quad+ \frac{1}{2} \left(  \sigma_{1}^{2} + 2 \rho\sigma_{1} \sigma_{2}
\pi+\sigma_{2}^{2} \pi^{2} \right)  \phi_{xx}(x, t) + \int_{0}^{\infty}(\phi(x
- \ell(z, t), t) - \phi(x, t)) \nu(\mathrm{d} z),
\end{align}
in which $\phi(x, t) \in C^{2, 1}(\mathbb{R} \times[0, T])$. In the first line
of \eqref{IG}, $u = (\ell, \pi)$, in which the constant $\pi$ represents the
strategy in which the insurer invests the constant $\pi$ in the risky asset.

\begin{theorem}
\label{thm:1} $($Verification theorem$)$. Suppose there exist $V(x, t)$ and
$g(x, t)\in C^{2, 1}(\mathbb{R} \times[0, T])$ satisfying the following conditions:

\item {$($1$)$} For all $(x, t) \in\mathbb{R} \times[0, T]$,
\begin{equation}
\label{EHJB1}\sup\limits_{\ell, \pi} \left\{  \mathcal{A}^{\ell, \pi} \, V(x, t) - \frac{\gamma}{2} \, \mathcal{A}^{\ell, \pi} \, g^{2}(x, t) + \gamma\,
g(x, t) \, \mathcal{A}^{\ell, \pi} \, g(x, t) \right\}  = 0.
\end{equation}
\hskip 18 pt Let $(\ell^{*}, \pi^{*})$ denote the pair that attains the supremum in \eqref{EHJB1}.

\item {$($2$)$} For all $(x, t) \in\mathbb{R} \times[0, T]$,
\begin{equation}
\label{EHJB2}\mathcal{A}^{\ell^{*}, \pi^{*}} g(x, t) = 0.
\end{equation}

\item {$($3$)$} For $x \in\mathbb{R}$,
\begin{equation}
\label{EHJB3}V(x, T) = x \text{ and } g(x, T) = x.
\end{equation}

Then, the equilibrium reinsurance-investment strategy $u^{*}$ is given by
\begin{equation}
\label{eq:ustar}u^{*}_{t} = \left(  \ell^{*}(Z_{t}, t), \pi^{*}(X_{t}, t)
\right)  .
\end{equation}
Note that $u^{*}$ is given in feedback form. $V(x, t) = J^{u^{*}}(x, t)$ is
the corresponding equilibrium value function, and $g(x, t) = \mathbb{E}_{x, t}
\left[  X_{T}^{u^{*}} \right]  $ is the expectation of terminal wealth.
\end{theorem}

\subsection{Equilibrium strategy}

One can use Theorem \ref{thm:1} directly to obtain \textit{an} equilibrium
strategy. However, we wish to show that there is only one such equilibrium
strategy; to that end, we have the following lemma, which is similar to Lemma
1 in Basak and Chabakauri \cite{BC10}.

\begin{lemma}
\label{lem:1} The value function $V$ and expectation of terminal wealth $g$
under the mean-variance criterion are separable in the surplus $x$ and admit
the following representation:
\begin{equation}
\label{conj}\left\{
\begin{array}
[c]{l}
V(x, t) = \mathrm{e}^{r(T - t)} \, x + B(t), \quad B(T) = 0,\\
g(x, t) = \mathrm{e}^{r(T - t)} \, x + b(t), \quad b(T) = 0.
\end{array}
\right.
\end{equation}

\end{lemma}

\begin{proof}
From \eqref{SDE3}, we have
\begin{align}
\label{SDE3_1}\mathrm{d} \left(  \mathrm{e}^{r(T - t)} \, X_{t}^{u} \right)
&  = \left[  (\mu- r) \pi_{t} + \int_{0}^{\infty}\left(  (\theta- \eta) z +
\eta\ell(z, t) \right)  \nu(\mathrm{d} z) \right]  \mathrm{e}^{r(T - t)} \,
\mathrm{d} t\nonumber\\
&  \quad+ \sqrt{\sigma_{1}^{2} + 2 \rho\sigma_{1} \sigma_{2} \pi_{t} +
\sigma_{2}^{2} \pi_{t}^{2}} \, \mathrm{e}^{r(T - t)} \, \mathrm{d} B_{t} -
\int_{0}^{\infty}\ell(z, t) \, \mathrm{e}^{r(T - t)} \, \tilde{N}(\mathrm{d}
z, \mathrm{d} t)\nonumber\\
&  =: G^{u}(t).
\end{align}
From \eqref{SDE3_1} it follows that
\begin{equation}
\label{eq:sep1}\mathbb{E}_{x, t} \left[  X^{u}_{T} \right]  = \mathrm{e}^{r(T
- t)} \, x + \mathbb{E}_{x, t} \left[  \int_{t}^{T} \left(  (\mu- r) \pi_{s} +
\int_{0}^{\infty}\left(  (\theta- \eta) z + \eta\ell(z, s) \right)
\nu(\mathrm{d} z) \right)  \mathrm{e}^{r(T - s)} \, \mathrm{d} s \right]  ,
\end{equation}
and
\begin{equation}
\label{eq:sep2}\mathrm{Var}_{x, t} \left[  X^{u}_{T} \right]  = \mathrm{Var}_{x, t} \left[  \int_{t}^{T} G^{u}(s) \, \mathrm{d} s \right]  .
\end{equation}
The expressions in \eqref{eq:sep1} and \eqref{eq:sep2} for the expectation and
variance of $X^{u}_{T}$, respectively, imply that the objective and the
expectation functions for the mean-variance criterion are separable in the
surplus $x$, as given in \eqref{conj}.
\end{proof}

In the next theorem, we present the equilibrium strategy and the corresponding
equilibrium value function.

\begin{theorem}
\label{thm:2}
The unique equilibrium reinsurance-investment strategy $u^{*} = (\ell^{*}(z, t), \pi^{*}(t))$ for the mean-variance criterion is given by
\begin{equation}
\label{equ:astar1}
\left\{
\begin{array}
[c]{l}
\ell^{*}(z, t) = \dfrac{\eta}{\gamma} \, \mathrm{e}^{-r(T-t)} \wedge z, \vspace{0.2cm}\\
\pi^{*}(t) = {\dfrac{\mu- r}{\gamma\, \sigma_{2}^{2}}} \, \mathrm{e}^{-r(T - t)} - \rho\, \dfrac{\sigma_{1}}{\sigma_{2}},
\end{array}
\right.
\end{equation}
and the corresponding value function is
\begin{equation}
\label{ovf}V(x, t) = \mathrm{e}^{r(T-t)} \, x + B(t),
\end{equation}
in which
\begin{align}
\label{eq:B}
B(t)  &= \int_{t}^{T} \left\{  \frac{1}{2 \gamma} \left( \frac{\mu- r}{\sigma_{2}} \right)  ^{2} + \mathrm{e}^{r(T - s)} \left[ - (\mu- r) \rho\, \frac{\sigma_{1}}{\sigma_{2}} + \int_{0}^{\infty}\left( (\theta- \eta) z + \eta\ell^{*}(z, s) \right)  \nu(\mathrm{d} z) \right]
\right. \nonumber\\
&  \qquad \qquad \left.  - \frac{\gamma}{2} \, \mathrm{e}^{2r(T - s)} \left[ \left(  1 - \rho^{2} \right)  \sigma_{1}^{2} + \int_{0}^{{\infty}} \left( \ell^{*}(z, s) \right)  ^{2} \nu(\mathrm{d} z) \right]  \right\}  \mathrm{d} s.
\end{align}
Furthermore,
\begin{equation}
\label{E1}
\mathbb{E}_{x, t}\left[  X_{T}^{u^{*}}\right]  = g(x, t) = \mathrm{e}^{r(T-t)} \, x + b(t),
\end{equation}
in which
\begin{equation}
\label{eq:b}
b(t) = \int_{t}^{T} \left\{  \frac{1}{\gamma} \left(  \frac{\mu- r}{\sigma_{2}} \right)  ^{2} + \mathrm{e}^{r(T - s)} \left[  - (\mu- r) \rho\, \frac{\sigma_{1}}{\sigma_{2}} + \int_{0} ^{\infty}\left(  (\theta- \eta) z + \eta\ell^{*}(z, s) \right)  \nu(\mathrm{d} z) \right]  \right\}  \mathrm{d} s.
\end{equation}

\end{theorem}

\begin{proof}
We verify that $u^*$, $V$, and $g$ defined, respectively, in \eqref{equ:astar1}, \eqref{ovf}, and \eqref{E1}, satisfy conditions (1)--(3) in Theorem \ref{thm:1}. To that end, from \eqref{IG}, we compute
\begin{equation}
\label{eq:A1}
\mathcal{A}^{\ell,\pi}\left(  \mathrm{e}^{r(T-t)}\,x + \beta(t) \right) = \beta_{t} + C^{\ell,\pi}(t) \,\mathrm{e}^{r(T-t)}, 
\end{equation}
in which $C$ is given by
\[
C^{\ell,\pi}(t)=(\mu-r)\pi+\int_{0}^{\infty}\left(  (\theta-\eta)z+\eta \,\ell(z,t)\right)  \nu(\mathrm{d}z).
\]
Also, from \eqref{IG}, we obtain, after simplifying,
\begin{align}\label{eq:A2}
&  \left(  \mathrm{e}^{r(T-t)}\,x+b(t)\right)  \mathcal{A}^{\ell,\pi}\left( \mathrm{e}^{r(T-t)}\,x+b(t)\right)  -\frac{1}{2}\,\mathcal{A}^{\ell,\pi
}\left(  \left(  \mathrm{e}^{r(T-t)}\,x+b(t)\right)  ^{2}\right) \nonumber \\
&  \quad=-\frac{1}{2}\,\mathrm{e}^{2r(T-t)}\left[  \left(  \sigma_{1}^{2}+2\rho\sigma_{1}\sigma_{2}\pi+\sigma_{2}^{2}\pi^{2}\right)  +\int_{0}^{\infty}\ell^{2}(z,t)\,\nu(\mathrm{d}z)\right]  .
\end{align}

By substituting the expressions for $V$ and $g$ from \eqref{conj} into \eqref{EHJB1} and by using the results of \eqref{eq:A1} and \eqref{eq:A2}, we get
\begin{equation}
\label{eq:EHJB1_1}
\sup\limits_{\ell, \pi} \left\{  B_{t} + C^{\ell, \pi}(t) \, \mathrm{e}^{r(T-t)} - \frac{\gamma}{2} \, \mathrm{e}^{2r(T - t)} \left[ \left(  \sigma_{1}^{2} + 2 \rho\sigma_{1} \sigma_{2} \pi+ \sigma_{2}^{2} \pi^{2} \right)  + \int_{0}^{\infty}\ell^{2}(z, t) \, \nu(\mathrm{d} z) \right]  \right\}  = 0.
\end{equation}
The expression in \eqref{eq:EHJB1_1} is concave with respect to $\pi$; thus, we obtain the optimal value of $\pi$ from the first-order condition. Specifically,
\begin{equation}
\label{eq:pistar}
\pi^{*}(t) = {\dfrac{\mu- r}{\gamma\, \sigma_{2}^{2}}} \, \mathrm{e}^{-r(T - t)} - \rho\, \dfrac{\sigma_{1}}{\sigma_{2}}.
\end{equation}
Next, consider the terms in $\ell$ in \eqref{eq:EHJB1_1}, that is,
\begin{equation}
\label{eq:int}
\int_{0}^{\infty}\left(  \eta\ell(z, t) - \frac{\gamma}{2} \, e^{r(T - t)} \, \ell^{2}(z, t) \right)  \nu(\mathrm{d} z).
\end{equation}
If we maximize the integrand in the integral in \eqref{eq:int} $z$-by-$z$ for a given $t \in[0, T]$, then we will maximize the integral itself. With respect to $\ell$, the graph of $f(\ell) := \eta\ell- \frac{\gamma}{2} \, e^{r(T - t)} \, \ell^{2}$ is a concave parabola that increases through the origin $(0, f(0)) = (0, 0)$; thus, $f$'s maximizer $\ell^{*} \in[0, z]$ is given by
\begin{equation}
\label{eq:ellstar}
\ell^{*}(z, t) = \dfrac{\eta}{\gamma} \, \mathrm{e}^{-r(T-t)} \wedge z.
\end{equation}

If we substitute $u^{*} = (\ell^{*}, \pi^{*})$ into \eqref{eq:EHJB1_1} and solve for $B(t)$ (by using the terminal condition $B(T) = 0$), then we obtain the expression in \eqref{eq:B}. Also, if we solve for $b(t)$ (with the same terminal condition $b(T) = 0$) in the equation $\mathcal{A}^{\ell^{*}, \pi^{*}} \left(  \mathrm{e}^{r(T-t)} \, x + b(t) \right)  = 0$, then we obtain the expression in \eqref{eq:b}.

Thus, $u^{*}$, $V$, and $g$ defined, respectively, in \eqref{equ:astar1}, \eqref{ovf}, and \eqref{E1}, satisfy conditions (1)--(3) in Theorem
\ref{thm:1}. To complete this proof, note that $u^{*}$ is an admissible strategy, as defined in Definition \ref{Def:2}.
\end{proof}

\begin{remark}
\label{rem:1} Lemma $\ref{lem:1}$ and Theorem $\ref{thm:2}$ prove that excess-loss reinsurance is the unique equilibrium strategy for a
time-consistent insurer under the mean-variance criterion; in that sense, we consider it optimal. Note that the equilibrium strategy is independent of the state variable $x$. This independence results from the fact that the risk aversion $\gamma$ is a constant. See \eqref{eq:mv} in which $\gamma$ in the mean-variance approximation of the utility's certainty equivalence represents the utility's absolute risk aversion. Recall that if utility exhibits constant absolute risk aversion, then the form of the utility function is exponential,
and decision making under exponential utility invariably results in strategies that are independent of the state variable. See Basak and Chabakauri \cite{BC10} and Bj\"{o}rk et al.\ \cite{BMZ14} for further discussion.

Moreover, the equilibrium excess-loss strategy is independent of the parameters of the risky asset and the safety loading of the insurer, while the equilibrium investment strategy is independent of the safety loadings of both the insurer and the reinsurer. In other words, the equilibrium reinsurance strategy is unaffected by the financial market, while the equilibrium investment strategy is unaffected by the price of reinsurance, and both strategies are unaffected by the price of the primary insurance.
\end{remark}

The behavior of the equilibrium strategy is given in the following corollary. The proof is straightforward and hence omitted. We discuss the intuition behind the behavior of the equilibrium strategy in the numerical analysis in the next section.

\begin{corollary}
The equilibrium retained claim $\ell^{*}(z, t)$ increases in $\eta$ and $t$, decreases in $r,$ $\gamma,$ and $T,$ and is independent of $x,$ $\theta,$ $\rho,$ $\sigma_{1},$ and $\sigma_{2};$ the equilibrium amounted invested in the risky asset $\pi^{*}(t)$ increases in $\mu$ and $t,$ decreases in $r,$ $\gamma,$ $\rho,$ and $T,$ and is independent of $x$ and $z$.
\end{corollary}

\subsection{Related problems}

In this section, we compare the equilibrium strategy in Theorem \ref{thm:2} with the optimal strategies for two related problems.

\subsubsection{Exponential utility}

First, as we observed in the Introduction and in Remark \ref{rem:1}, the mean-variance criterion is related to maximizing expected utility of terminal wealth under \textit{constant} absolute risk aversion $\gamma$.  For the latter problem, a standard verification theorem states that if we find a classical solution $U$ to $\sup_{\ell, \pi} \mathcal{A}^{\ell, \pi} U(x, t) = 0$, with terminal condition $U(x, T) = - e^{-\gamma x}$, then $U$ equals $\sup_{\ell, \pi} \mathbb{E}_{x, t} \left[ u \left(  X^{u}_{T} \right)  \right]  $, in which $u(x) = - e^{-\gamma x}$. Furthermore, the optimal strategy is given in feedback form by the maximizer of $\mathcal{A}^{\ell, \pi} U(x, t)$.

For the model in this paper, it is straightforward to show that the optimal strategy is $(\ell^{u}, \pi^{u})$, in which $\ell^{u}$ and $\pi^{u}$ are given by
\begin{equation}
\label{equ:ustar1}
\left\{
\begin{array}
[c]{l}
\ell^{u}(z, t) = \dfrac{\ln(1 + \eta)}{\gamma} \, \mathrm{e}^{-r(T-t)} \wedge z, \vspace{0.2cm}\\
\pi^{u}(t) = {\dfrac{\mu- r}{\gamma\, \sigma_{2}^{2}}} \, \mathrm{e}^{-r(T - t)} - \rho\, \dfrac{\sigma_{1}}{\sigma_{2}}.
\end{array}
\right.
\end{equation}
Note that, for small values of $\eta$, $\ell^{u}$ is approximately equal to $\ell^{*}$, but $\pi^{u}$ is identically equal to $\pi^{*}$. This result further confirms the close relationship between finding the equilibrium strategy for the mean-variance criterion with constant risk aversion parameter $\gamma$ and maximizing expected utility of terminal wealth with constant absolute risk aversion $\gamma$.

\subsubsection{Pre-commitment}

Second, if the insurer pre-commits to its strategy at time 0 for the entire period $[0, T]$ to maximize the time-0 mean-variance objective function in \eqref{Jdef}, then the optimal {\it investment} strategy differs from $\pi^*$, as shown in Basak and Chabakauri \cite{BC10}.  Furthermore, the optimal {\it reinsurance} strategy also differs from $\ell^*$.  We demonstrate the latter statement in this section.

The pre-commitment problem is given by
\begin{equation}
\label{eq:prec}
\sup_{\ell, \pi} \, \mathbb{E}_{x, 0} \left[ X_{T} \right] - \frac{\gamma}{2} \, \mathrm{Var}_{x, 0}  \left[ X_{T} \right].
\end{equation}
By following the work in Zhou and Li \cite{ZL00}, we first solve the following auxiliary problem
\begin{equation}
\label{eq:aux}
\mathcal{U}(x, t) = \sup_{\ell, \pi} \mathbb{E}_{x, t} \left[ \alpha X_T - \frac{\gamma}{2} \, X_T^2 \right],
\end{equation}
with the optimal strategy given in feedback form by $\left( \hat \pi(\alpha, X_t, t), \hat \ell(\alpha, z, X_t, t) \right)$.  Then, by setting $\alpha$ equal to the solution $\alpha^*$ of the following equation
\[
\alpha = 1 + \gamma \mathbb{E}_{x_0, 0} \left( X^{\hat \pi(\alpha, X_t, t), \hat \ell(\alpha, z, X_t, t)}_T \right),
\]
$(\hat \pi, \hat \ell)$ with $\alpha = \alpha^*$ equals the optimal strategy for the pre-commitment problem in \eqref{eq:prec}.  We anticipate the control $\ell$ will depend on the state variable $x$ and write $\ell = \ell(z, x, t)$ in feedback form.  Furthermore, $(\hat \pi, \hat \ell)$ clearly depends on $x_0$ through $\alpha^*$.  Note that $\mathcal{U}$ in \eqref{eq:aux} is concave with respect to $x$ because $\alpha x - \frac{\gamma}{2} \, x^2$ is concave and the surplus is linear with respect to the controls.

If we find a classical solution $\sup_{\ell, \pi} \mathcal{A}^{\ell, \pi} \mathcal{V}(x, t) = 0$, with terminal condition $\mathcal{V}(x, T) = \alpha x - \frac{\gamma}{2} \, x^2$, then $\mathcal{V} = \mathcal{U}$, the value function of the auxiliary problem in \eqref{eq:aux}.    Suppose we have such a classical solution of this boundary-value problem; without ambiguity, write it as $\mathcal U$.  Then, the terms in the Hamilton-Jacobi-Bellman equation involving $\ell$ are
\[
\max_{\ell} \, \int_0^\infty \left( (1 + \eta) \ell \, \mathcal{U}_x(x, t) + \mathcal{U}(x - \ell, t) \right) \nu(\drm z).
\]
As in the proof of Theorem \ref{thm:2}, we maximize the integral by maximizing the integrand $z$-by-$z$ for a fixed value of $(x, t)$ over $\ell$ such that $0 \le \ell(z, x, t) \le z$.  Because $\mathcal{U}$ is concave with respect to $x$, it is straightforward to show that the optimal reinsurance is of the form
\begin{equation}
\hat \ell(z, x, t) = d(x, t) \wedge z,
\end{equation}
in which $d = d(x, t)$ is given by
\begin{equation}
\label{equ:ustar2}
d(x, t) = 
\left\{
\begin{array}
[c]{ll}
\ell_c, &\mathrm{if} \; \exists \ell_c \in (0, z) \; \mathrm{s.t.} \; (1 + \eta)\mathcal{U}_x(x, t) = \mathcal{U}_x(x - \ell_c, t), \vspace{0.2cm}\\
\infty, &\mathrm{if} \; (1 + \eta)\mathcal{U}_x(x, t) - \mathcal{U}_x(x - \ell_c, t) > 0, \; \forall \ell > 0 \vspace{0.2cm}\\
0, &\mathrm{if} \;\, \mathcal{U}_x(x, t) \le 0.
\end{array}
\right.
\end{equation}

At time $T$, $\mathcal{U}(x, T) = \alpha^* x - \frac{\gamma}{2} \, x^2$, and $\alpha^*$ depends on $x_0$; thus,
\[
\hat \ell(z, x, T) = \eta \left( \frac{\alpha^*}{\gamma} - x \right)_+ \wedge z \ne \frac{\eta}{\gamma} \wedge z = \ell^*(z, T).
\]
Thus, the optimal pre-commitment reinsurance strategy differs from the equilibrium reinsurance strategy for the time-consistent problem.\footnote{As an aside, note that if we begin the pre-commitment problem at time $T - \epsilon$ for $\epsilon > 0$ small, then $\lim_{\epsilon \to 0} \hat \ell(z, x, T) = \frac{\eta}{\gamma} \wedge z = \ell^*(z, T)$ because $\lim_{\epsilon \to 0} x_{T-\epsilon} = x$, in which $X_T = x$.  This equality makes sense because pre-committing over a vanishingly small interval is equivalent to being time-consistent over that interval.}

\section{Numerical examples}

\begin{example}
[Equilibrium strategies] In this example, we examine the sensitivity of the equilibrium reinsurance-investment strategies given in \eqref{equ:astar1} to different parameters. Unless otherwise stated, the parameter values are given by $r = 0.05,$ $\mu = 0.10,$ $\sig_{1} = 0.20,$ $\sig_{2} = 0.30,$ $\eta = 0.60,$ $\rho = 0.50,$ $\gamma = 1,$ and $T = 9$. Denote the corresponding equilibrium strategy by $(m^*, \pi^*),$ in which
\[
m^*(t)=\dfrac{\eta}{\gamma}\,\mathrm{e}^{-r(T-t)}\,.
\]
\[
{\parbox[b]{3.5284in}{\begin{center}
\includegraphics[
height=2.6498in,
width=3.5284in
]
{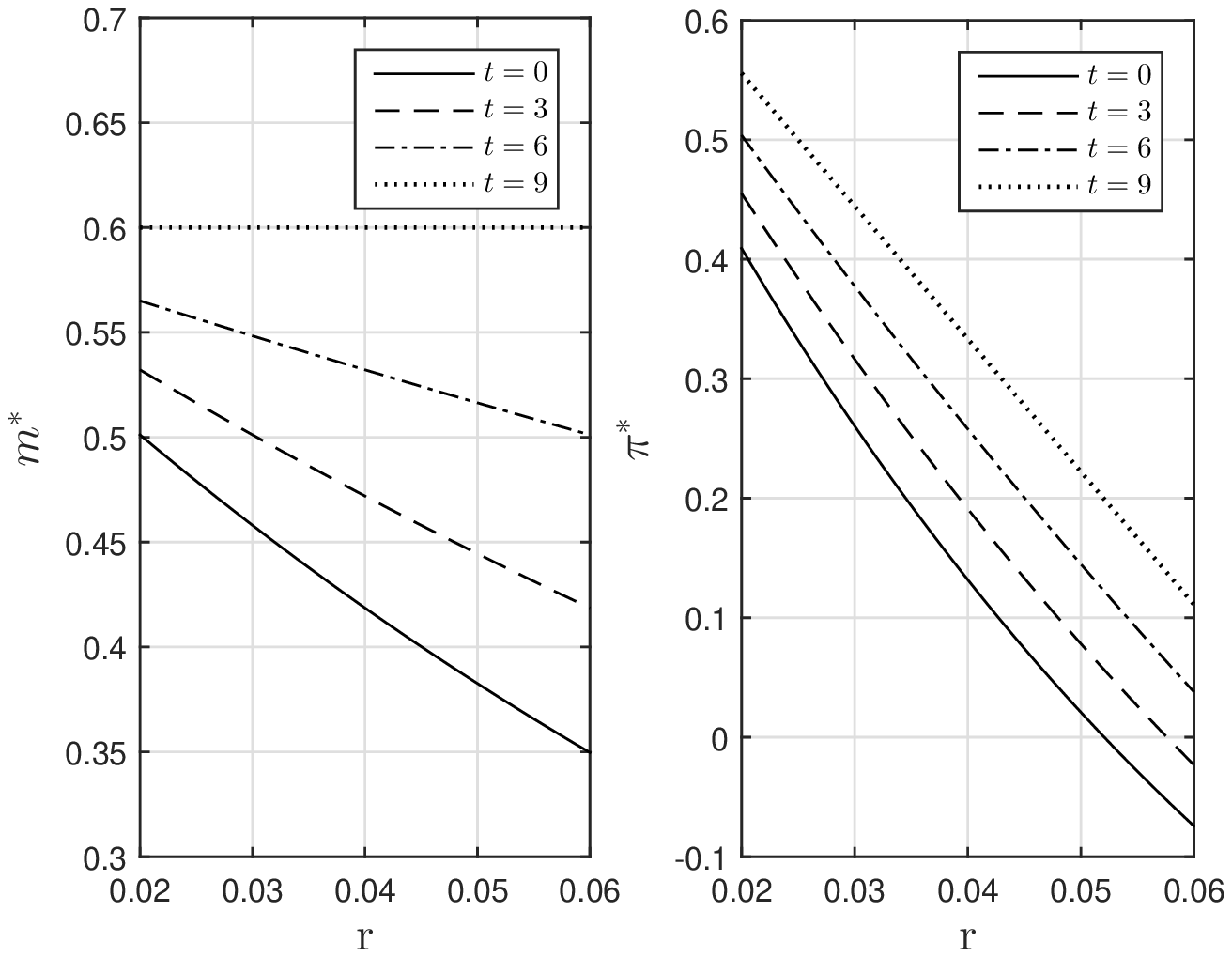}
\\
Figure 1. Impact of $r.$
\end{center}}}
\]
\[
{\parbox[b]{3.5284in}{\begin{center}
\includegraphics[
height=2.6498in,
width=3.5284in
]
{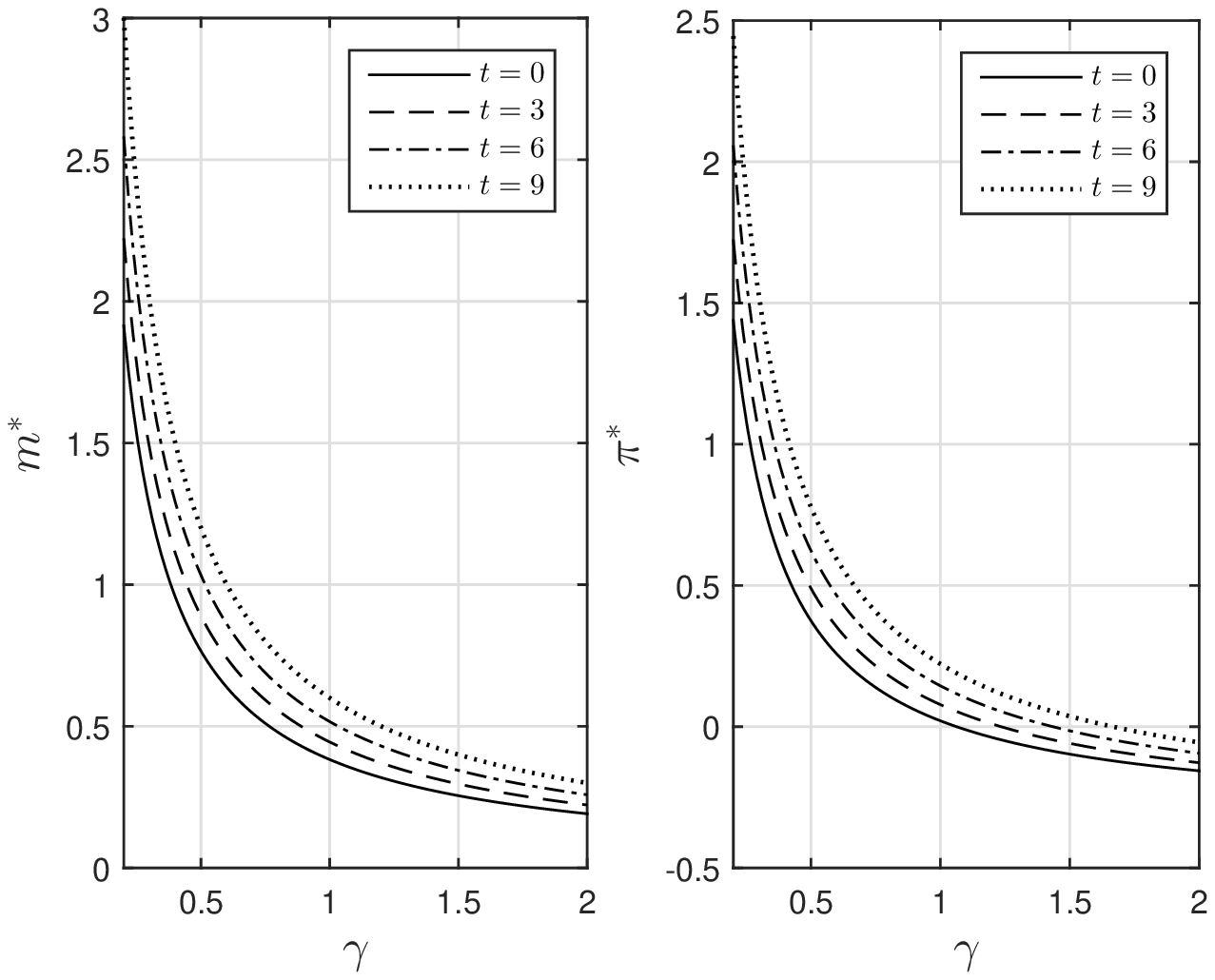}
\\
Figure 2. Impact of $\gamma.$
\end{center}}}
\]

In Figure 1, we plot the impact of $r$ on the reinsurance-investment strategy for a variety of times $t$. Both $m^*$ and $\pi^*$ decrease as the risk-free rate increases, except for $m^*$ when $t=T$, at which it is constant. When large claims occur, the insurer might borrow from the risk-free asset to aid in regaining solvency; recall that the amount invested in the risk-free asset equals $x-\pi^*(t)$, which is negative when the surplus $x$ is negative. Thus, as borrowing money becomes more costly, the insurer retains less insurance risk. Furthermore, it is reasonable for the insurer to decrease the amount invested in the risky asset as the risk-free asset becomes more attractive.

In Figure 2, we plot the impact of $\gamma$ on the reinsurance-investment strategy. Note that, as the insurer becomes more risk averse, it assumes less insurance risk and less financial risk.

\[
{\parbox[b]{3.5284in}{\begin{center}
\includegraphics[
height=2.6498in,
width=3.5284in
]
{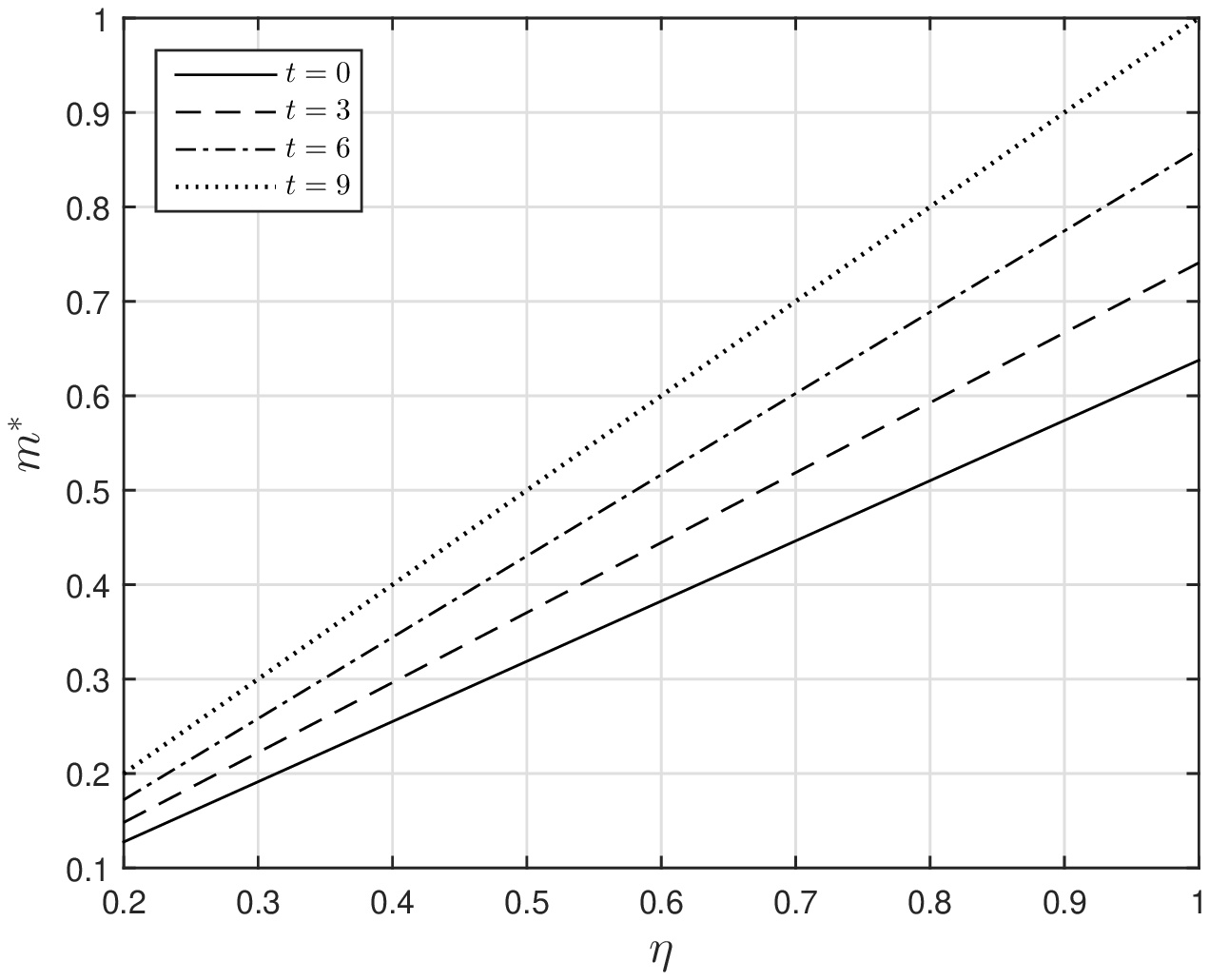}
\\
Figure 3. Impact of $\eta$ on $m^*$.
\end{center}}}
\]
\[
{\parbox[b]{6.9788in}{\begin{center}
\includegraphics[
trim=0.986641in 0.000000in 0.000000in 0.000000in,
height=3.1377in,
width=6.9788in
]
{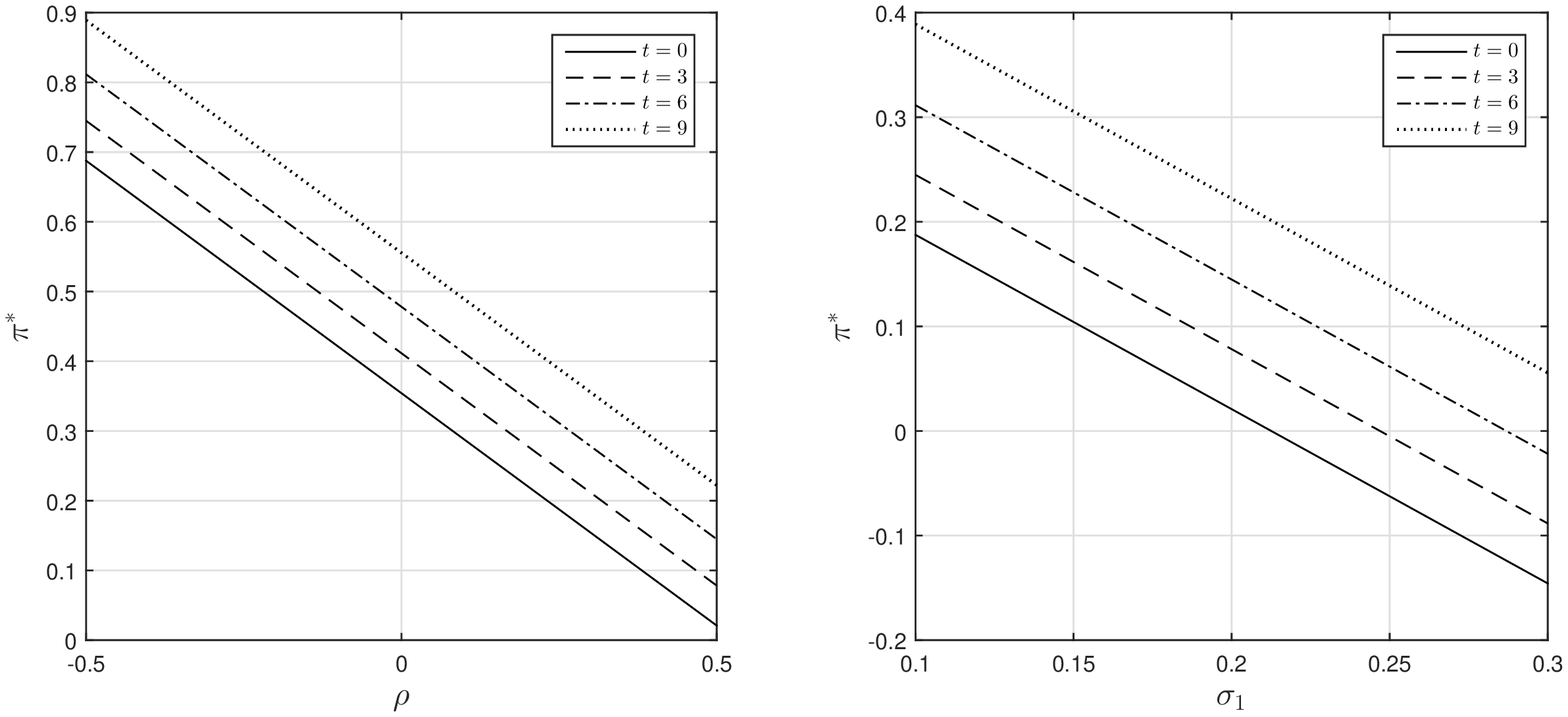}
\\
Figure 4. Impact of $\rho$ $($left$)$ and $\sigma_1$ $($right$)$ on $\pi^*$.
\end{center}}}
\]

In Figure 3, we plot the impact of $\eta$ on the retention level $m^{*}$. Note that $m^{*}$ increases as $\eta$ increases. In other words, as the reinsurance policy becomes more expensive, the insurer retains more insurance risk.

In Figure 4, we plot the impact of $\rho,$ $\sigma_{1}$ on the investment strategy. First, we see from the left panel that as $\rho$ increases, $\pi^*$ decreases. Second, we see from the right panel that, as the insurance market becomes more volatile, the amount invested in the financial market decreases because there is a positive correlation $(\rho = 0.50)$ between the two markets.
\end{example}

\begin{example}
[Proportional vs.\ excess-loss reinsurance] In this example, we assume that the basic surplus process follows the classical Cram\'{e}r-Lundberg model
\[
\mathrm{d}U_t = c\,\mathrm{d}t - \mathrm{d} \sum_{i=1}^{N_{t}}Y_{i}, \quad U_{0}=u,
\]
in which $\left\{ Y_{i} \right\}_{i=1}^{\infty}$ is a sequence of independent and identically distributed exponential random variables with
common survival function $S(y):= \mathrm{e}^{-\kappa y}$ for $y > 0$ representing the amount of individual claims, and $\left\{  N_{t}\right\}_{t\geq0}$ is a Poisson process with intensity $\lambda > 0$ representing the number of claims, independent of $\left\{  Y_{i}\right\}  $. The premium rate $c = (1 + \theta)\frac{\lambda}{\kappa}$. By applying equation \eqref{ovf} with $\nu (\mathrm{d}z)=\lambda F(\mathrm{d}z),$ $\sigma_{1}=0,$ and $\sigma_{2} = \sigma,$ the corresponding value function under the Cram\'{e}r-Lundberg model is given by
\[
V_{1}(x,t) = \mathrm{e}^{r(T-t)}\,x+B_{1}(t),\quad (x,t) \in \mathbb{R} \times [0, T],
\]
in which
\begin{align*}
B_{1}(t)  &  =\int_{t}^{T}\left\{  \frac{1}{2\gamma}\left(  \frac{\mu-r}{\sigma}\right)^{2}+\mathrm{e}^{r(T-s)}\left[  (\theta-\eta
)\lambda\mathbb{E}[Y]+\eta\lambda\int_{0}^{\frac{\eta}{\gamma}\mathrm{e}^{-r(T-s)}}S(y)\,\mathrm{d}y\right]  \right. \\
&  \qquad\qquad\left.  -\gamma\lambda\mathrm{e}^{2r(T-s)}\int_{0}^{\frac{\eta}{\gamma}\mathrm{e}^{-r(T-s)}}yS(y)\,\mathrm{d}y\right\}  \mathrm{d}s.
\end{align*}
Moreover, we have
\[
g_{1}(x,t)=\mathbb{E}_{1}^{x,t}\left[  X_{T}^{u^*}\right]  =\mathrm{e}^{r(T-t)}\,x+b_{1}(t),\quad(x,t)\in\mathbb{R}\times\lbrack0,T],
\]
in which
\[
b_{1}(t)={\int}_{t}^{T}\left\{  \frac{1}{\gamma}\left(  \frac{\mu-r}{\sigma}\right)  ^{2}+\mathrm{e}^{r(T-s)}\left[  (\theta-\eta)\lambda\mathbb{E}
[Y]+\eta\lambda\int_{0}^{\frac{\eta}{\gamma}\mathrm{e}^{-r(T-s)}}S(y)\,\mathrm{d}y\right]  \right\}  \mathrm{d}s,
\]
and $\mathrm{Var}_{1}^{x,t}(X_{T}^{u^*})=\frac{2}{\gamma}\left(g_{1}(x,t) - V_{1}(x,t)\right)$. The other parameter values equal $r=0.05,$ $\mu=0.10,$ $\sigma=0.30,$ $\gamma=0.50,$ $T=3,$ $\theta=0.50,$ $\eta=0.60,$ $\lambda=1,$ and $\kappa=0.50$.

Under this model, we compare the value function $V_{1}$ with the value function under the equilibrium proportional reinsurance $V_{2}$ determined in Zeng et al.\ \cite{ZLL13}, for example. We see from the top panel in Figure 5 that $V_{1}$ dominates $V_{2}$ except at the boundary $t=T,$ where $V_{1}(x,T)=V_{2}(x,T)=x$. In other words, the equilibrium proportional reinsurance policy is demonstrably not optimal within a larger class, that is, the broad class of reinsurance policies $\ell$ for which $0\leq \ell (Z_{t}, t) \leq Z_{t}$. We also see from the bottom panel in Figure 5 that when mean and variance are viewed separately, compared to the equilibrium proportional reinsurance, though the equilibrium excess-loss policy generates a greater terminal mean, the associated terminal risk is also greater.

\[
{\includegraphics[
height=2.6498in,
width=3.5284in
]
{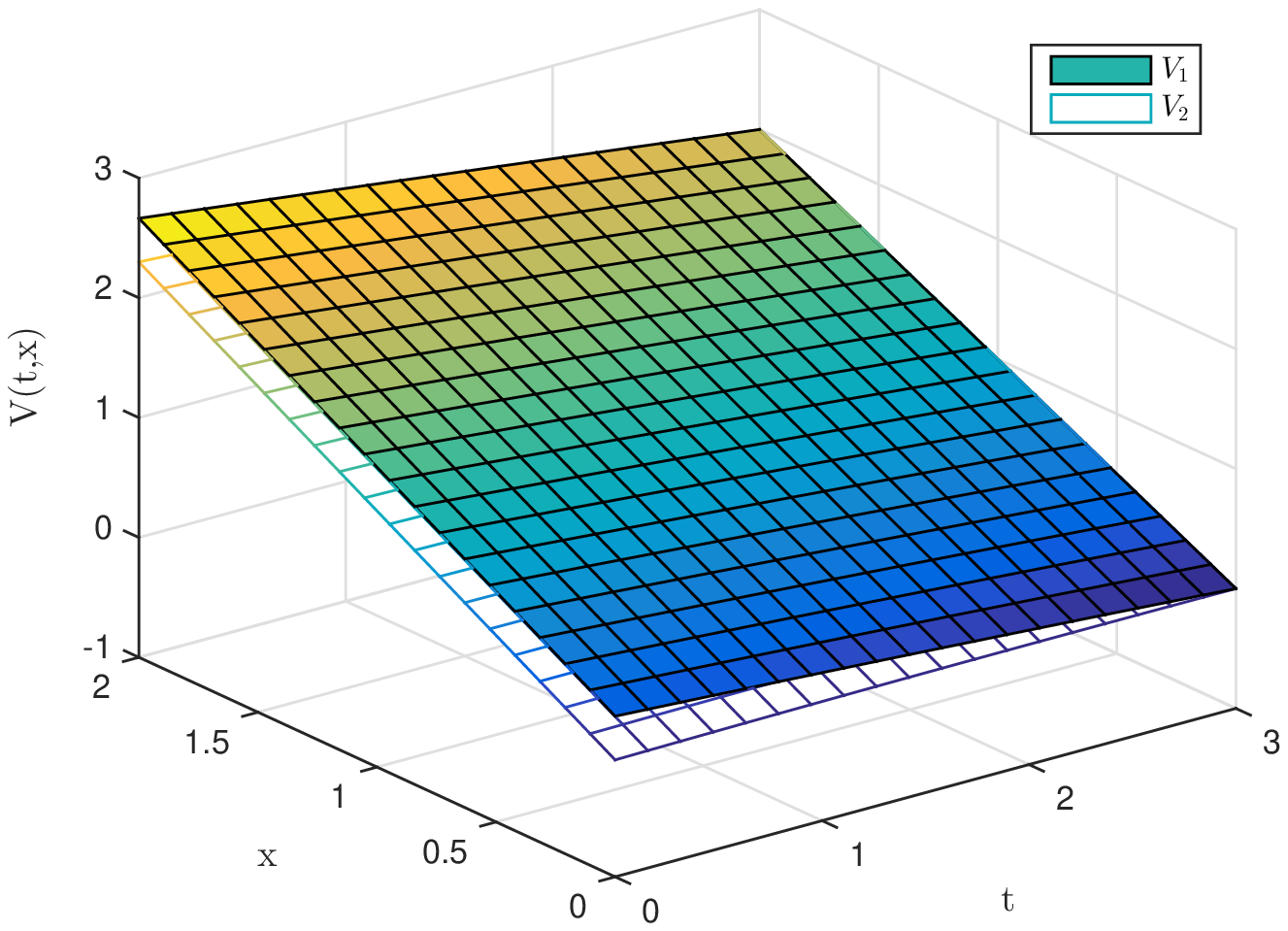}
}
\]

\vspace{-0.8in}
\[
{\parbox[b]{5.3688in}{\begin{center}
\includegraphics[
height=2.7769in,
width=5.3688in
]
{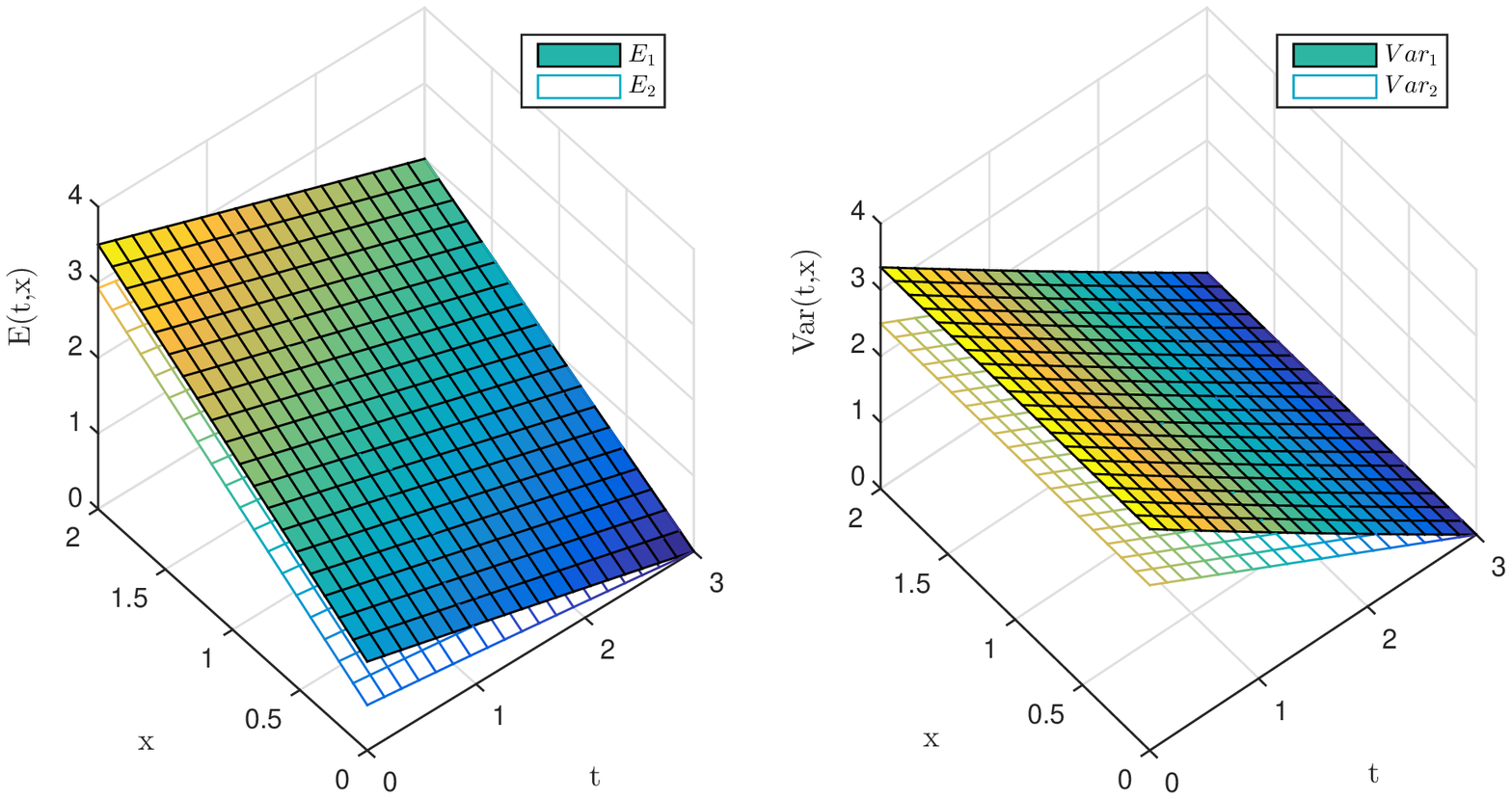}
\\
Figure 5. Excess-loss vs.\ proportional reinsurance.
\end{center}}}
\]

\end{example}

\section{Future research}

In future research, we will extend the work in this paper in two directions.  First, we will allow for the coefficient of absolute risk-aversion to depend on the surplus, as in Bj\"ork et al.\ \cite{BMZ14}.  Risk aversion is generally considered to decrease with wealth, so it would be reasonable to choose $\gamma(x)$ to be a decreasing function of $x$.  One natural choice for $\gamma(x) = \frac{\delta}{x}$, in which $\delta$ is a positive constant, which one can interpret as the {\it relative} risk aversion.

Second, we will consider premium principles other than the expected value premium principle for the reinsurer.  For example, we do not necessarily expect excess-loss reinsurance to be optimal for the standard deviation, variance, or Wang premium principle.

\end{document}